\documentclass{llncs}

\usepackage{oldlfont,amsmath,amssymb,epsf,stmaryrd,epsfig,color}

\usepackage{verbatim}

\def\fa{\forall}
\def\ex{\exists}

\newcommand{\VL}[1]{}
\newcommand{\VC}[1]{#1}

\newbox\tempa
\newbox\tempb
\newdimen\tempc
\def\mud#1{\hfil $\displaystyle{\mathstrut #1}$\hfil}
\def\rig#1{\hfil $\displaystyle{#1}$}
\def\irulehelp#1#2#3{\setbox\tempa=\hbox{$\displaystyle{\mathstrut #2}$}%
                        \setbox\tempb=\vbox{\halign{##\cr
        \mud{#1}\cr
        \noalign{\vskip\the\lineskip}
        \noalign{\hrule height 0pt}
        \rig{\vbox to 0pt{\vss\hbox to 0pt{${\; #3}$\hss}\vss}}\cr
        \noalign{\hrule}
        \noalign{\vskip\the\lineskip}

        \mud{\copy\tempa}\cr}}
                      \tempc=\wd\tempb
                      \advance\tempc by \wd\tempa
                      \divide\tempc by 2 }
\def\irule#1#2#3{{\irulehelp{#1}{#2}{#3}
                     \hbox to \wd\tempa{\hss \box\tempb \hss}}}

\begin{document}
\title{Decidability, Introduction Rules, and Automata}
\author{Gilles Dowek\inst{1} \and Ying Jiang\inst{2}}
\institute{
Inria, 
23 avenue d'Italie,
CS 81321, 75214 Paris Cedex 13, France, 
{\tt gilles.dowek@inria.fr}.
\and
State Key Laboratory of Computer Science,
Institute of Software, 
Chinese Academy of Sciences,
100190 Beijing, China,
{\tt jy@ios.ac.cn}.}
\date{}
\maketitle
\pagestyle{plain}

\begin{abstract}
We present a method to prove the decidability of provability in
several well-known inference systems. This method generalizes both
cut-elimination and the construction of an automaton recognizing the
provable propositions.
\end{abstract}

\section{Introduction}

The goal of this paper is to connect two areas of logic: proof theory
and automata theory, that deal with similar problems, using a different
terminology.

To do so, we first propose to unify the terminology, by extending the
notions of {\em introduction rule}, {\em automaton}, {\em cut}, and
{\em cut-elimination} to arbitrary inference systems.  An {\em
  introduction rule} is defined as any rule whose premises are smaller
than its conclusion and an {\em automaton} as any inference system
containing introduction rules only. Provability in an automaton is
obviously decidable.  A {\em cut} is defined as any proof ending with
a non-introduction rule, whose major premises are proved with a proof 
ending with introduction rules.  We show that a cut-free proof
contains introduction rules only.  A system is said to have the {\em
cut-elimination property} if every proof can be transformed into a
cut-free proof. Such a system is equivalent to an automaton.

Using this unified terminology, we then propose a general {\em
saturation} method to prove the decidability of an inference
system, by transforming it into a system that has the cut-elimination
property, possibly adding extra rules. The outline of this method is
the following. Consider a proof containing a non-introduction rule and
focus on the sub-proof ending with this rule
$${\small 
  \irule{\irule{\pi^1}
               {s^1}
               {}
        ~~~...~~~
         \irule{\pi^n}
               {s^n}
               {}
        }
        {s}
        {\mbox{non-intro}}}$$
Assume it is possible to recursively eliminate the cuts 
in the proofs $\pi^1$, ..., $\pi^n$, that is to transform them into 
proofs containing introduction rules only, hence 
ending with an introduction rule. We obtain a proof of the form
$${\small 
  \irule{\irule{\irule{\rho^1_1}{s^1_1}{}~~...~~\irule{\rho^1_{m_1}}{s^1_{m_1}}{}}
               {s^1}
               {\mbox{intro}}
         ~~~~~~~~~~~~~~~~~~...~~~~~~~~~~
         \irule{\irule{\rho^n_1}{s^n_1}{}~~...~~\irule{\rho^n_{m_n}}{s^n_{m_n}}{}}
               {s^n}
               {\mbox{intro}}
        }
        {s}
        {\mbox{non-intro}}}$$
We may moreover tag each premise $s^1$, ..., $s^n$
of the non-introduction rule as {\em major} or {\em minor}. 
For 
instance, each elimination rule of
Natural Deduction \cite{Prawitz} has one major premise and 
the cut rule of Sequent Calculus \cite{Kleene} has two. 
If the major premises are $s^1$, ..., $s^k$ and {\em minor} ones 
$s^{k+1}, ..., s^n$, the proof above can be 
decomposed as
$${\small 
  \irule{\irule{\irule{\rho^1_1}{s^1_1}{}~~...~~\irule{\rho^1_{m_1}}{s^1_{m_1}}{}}
               {s^1}
               {\mbox{intro}}
         ~~~~~~~~~~~~~~~~~~...~~~~~~~~~~
         \irule{\irule{\rho^k_1}{s^k_1}{}~~...~~\irule{\rho^k_{m_k}}{s^k_{m_k}}{}}
               {s^k}
               {\mbox{intro}}
         ~~~~~~~~~~~~~~~~
         \irule{\pi'^{k+1}}
               {s^{k+1}}
               {}
         ~~...~~
         \irule{\pi'^n}
               {s^n}
               {}
        }
        {s}
        {\mbox{non-intro}}}$$
A proof of this form is called a {\em cut} and it must be reduced 
to another proof. The definition of the reduction is specific to 
each system under consideration. In several cases, however, such a cut 
is reduced to a proof built with the proofs 
$\rho^1_1$, ...,
$\rho^1_{m_1}$, ..., $\rho^k_1$, ..., $\rho^k_{m_k}$, $\pi'^{k+1},$ ..., 
$\pi'^n$
and 
a derivable rule allowing to deduce the conclusion $s$ from the
premises $s^1_1, ..., s^1_{m_1}, ...,$ $s^k_1, ..., s^k_{m_k}, s^{k+1},
..., s^n$. Adding such derivable rules in order to eliminate cuts is 
called a {\em saturation} procedure. 

Many cut-elimination proofs, 
typically the cut-elimination proofs for Sequent Calculus 
\cite{GirardLafontTaylor}, 
do not proceed by eliminating cuts step by
step, but by proving that a non-introduction rule is admissible in the
system obtained by dropping this rule, that is, proving that
if the premises $s^1, ..., s^n$ of this rule are
provable in the restricted system, 
then so is its conclusion $s$.  Proceeding by induction on
the structure of proofs of $s^1, ..., s^n$ leads to consider cases
where each major premise $s^i$ has a proof ending with an introduction
rule, that is also proofs of the form
$${\small 
  \irule{\irule{\irule{\rho^1_1}{s^1_1}{}~~...~~\irule{\rho^1_{m_1}}{s^1_{m_1}}{}}
               {s^1}
               {\mbox{intro}}
         ~~~~~~~~~~~~~~~~~~...~~~~~~~~~~
         \irule{\irule{\rho^k_1}{s^k_1}{}~~...~~\irule{\rho^k_{m_k}}{s^k_{m_k}}{}}
               {s^k}
               {\mbox{intro}}
          ~~~~~~~~~~~~~~~~~
         \irule{\pi_{k+1}}{s^{k+1}}{}
         ~~...~~
         \irule{\pi_n}{s^n}{}
        }
        {s}
        {\mbox{non-intro}}}$$

In some cases, the saturation method succeeds showing that every proof
can be transformed into a proof formed with introduction rules
only. Then, the inference system under consideration is equivalent,
with respect to provability, to the automaton obtained by dropping all
its non-introduction rules.  This equivalence obviously ensures the
decidability of provability in the inference system.  In other cases, in
particular when the inference system under consideration is
undecidable, the saturation method succeeds only partially: typically
some non-introduction rules can be eliminated but not all, or only a
subsystem is proved to be equivalent to an automaton.

This saturation method is illustrated with examples coming from both
proof theory and automata theory: Finite Domain Logic, Alternating
Pushdown Systems, and three fragments of Constructive Predicate Logic,
for which several formalizations are related: Natural Deduction,
Gentzen style Sequent Calculus, Kleene style Sequent Calculus, and
Vorob'ev-Hudelmaier-Dyckhoff-Negri style Sequent Calculus.  The
complexity of these provability problems, when they are decidable, is
not discussed in this paper and is left for future work, for instance
in the line of \cite{BasinGanzinger,McAllester}.

In the remainder of this paper, the notions of introduction rule,
automaton, and cut are defined in Section \ref{secautomata}.  Section
\ref{seclabels} discusses the case of Finite State Automata.  In
Sections \ref{secdecidable} and \ref{secpartial}, examples of
cut-elimination results are presented.  In the examples of Section
\ref{secdecidable}, the non-introduction rules can be completely
eliminated transforming the inference systems under considerations
into automata, while this elimination is only partially successful in
the undecidable examples of Section \ref{secpartial}. 
\VC{The proofs, and
some developments, are omitted from this extended abstract.  They can
be found in the long version of the paper
{\tt https://who.rocq.inria.fr/Gilles.Dowek/Publi/introlong.pdf}}.

\section{Introduction rules, Automata, and Cuts}
\label{secautomata}

\subsection{Introduction rules and Automata}

Consider a set $S$, whose elements typically are propositions,
sequents, etc. Let $S^*$ be the set of finite lists of elements of 
$S$. 

\begin{definition}[Inference rule, Inference system, Proof]
An {\em inference rule} is a partial function from $S^*$ to $S$.  
If $R$ is an inference rule and $s = R(s_1, ..., s_n)$, 
we say that the conclusion $s$ is {\em proved} from the premises
$s_1, ..., s_n$ with the rule $R$ and
we write
$${\small \irule{s_1~...~s_n}{s}{R}}$$
Some rules are equipped with an
extra piece of information, tagging each premise $s_1$, ..., $s_n$
as {\em major} or {\em minor}. 
An {\em inference system} is a
set of inference rules.  A {\em proof} in an inference system is a
finite tree labeled by elements of $S$ such that for each node labeled
with $s$ and whose children are labeled with $s_1$, ..., $s_n$, there
exists an inference rule $R$ of the system such that
$${\small \irule{s_1~...~s_n}{s}{R}}$$
A proof is a {\em proof} of $s$ if its root is labeled by $s$.  An
element of $S$ is said to be {\em provable}, if it has a proof.
\end{definition}

\begin{definition}[Introduction rule, Pseudo-automaton] 
Consider a set $S$ and a well-founded 
order $\prec$ on $S$. 
A rule $R$ is said to
be an {\em introduction} rule with respect to this order, if whenever
$${\small \irule{s_1~...~s_n}
        {s}
        {R}}$$
we have $s_1 \prec s$, ..., $s_n \prec s$.  A {\em pseudo-automaton}
is an inference system containing introduction rules only.
\end{definition}

Except in the system ${\cal D}$ (Section \ref{secdychkoff}), this
order $\prec$ is always that induced by the size of the propositions
and sequents. It is left implicit.

\begin{definition}[Finitely branching system, Automaton]
\label{defautomaton}
An inference system is said to be {\em finitely branching}, if for
each conclusion $s$, there is only a finite number of lists of
premises $\overline{s}_1$, ..., $\overline{s}_p$ such that $s$ can be
proved from $\overline{s}_i$ with a rule of the system.  An {\em
automaton} is a finitely branching pseudo-automaton.
\end{definition}

\subsection{Cuts}

We define a general notion of cut, that applies to all inference
systems considered in this paper.  More specific notions of cut
will be introduced later for some systems, and the general notion of
cut defined here will be emphasized as {\em general cut} to avoid
ambiguity.

\begin{definition}[Cut]
\label{defgeneralcut}
A {\em (general) cut} is a proof of the form 
$${\small 
  \irule{\irule{\irule{\rho^1_1}{s^1_1}{}~~...~~\irule{\rho^1_{m_1}}{s^1_{m_1}}{}}
               {s^1}
               {\mbox{intro}}
         ~~~~~~~~~~~~~~~~~~...~~~~~~~~~~
         \irule{\irule{\rho^k_1}{s^k_1}{}~~...~~\irule{\rho^k_{m_k}}{s^k_{m_k}}{}}
               {s^k}
               {\mbox{intro}}
          ~~~~~~~~~~~~~~~~~
         \irule{\pi_{k+1}}{s^{k+1}}{}
         ~~...~~
         \irule{\pi_n}{s^n}{}
        }
        {s}
        {\mbox{non-intro}}}$$
where $s^1, ..., s^k$ are the major premises of the non-introduction
rule.  A proof {\em contains a cut} if one of its sub-proofs is a cut.
A proof is {\em cut-free} if it contains no cut.  An inference system
{\em has the cut-elimination property} if every element that has a
proof also has a cut-free proof.
\end{definition}

\begin{lemma}[Key lemma]
\label{key}
A proof is cut-free if and only if it contains introduction rules only. 
\end{lemma}

\begin{proof}
If a proof contains introduction rules only, it is obviously cut-free.
We prove the converse by induction over proof structure.  Consider a
cut-free proof. Let $R$ be the last rule of this proof and $\pi_1$,
..., $\pi_n$ be the proofs of the premises of this rule. The proof has
the form
$${\small 
  \irule{\irule{\pi_1}
               {s_1} 
               {}
         ~~~...~~~
         \irule{\pi_n}
               {s_n}
               {}
        }
        {s}
        {R}}$$
By induction hypothesis, the proofs $\pi_1$, ..., $\pi_n$ contain
introduction rules only. As the proof is cut-free, the rule $R$ must
be an introduction rule.
\end{proof}

Consider a finitely-branching 
inference system ${\cal I}$ and the automaton ${\cal A}$
formed with the introduction rules of ${\cal I}$. If ${\cal I}$ has
the cut-elimination property, then every element that has a proof in
${\cal I}$ has a cut-free proof, that is a proof formed with
introduction rules of ${\cal I}$ only, that is a proof in ${\cal
  A}$.  Thus, ${\cal I}$ and ${\cal A}$ are equivalent with respect to
provability. Since ${\cal A}$ is decidable, so is ${\cal I}$.

\section{Finite State Automata}
\label{seclabels}

In this section, we show that the usual notion of finite state
automaton is a particular case of the notion of automaton introduced
in Definition \ref{defautomaton}.

\VL{\subsection{Representing Finite State Automata}}

Consider a finite state automaton ${\cal A}$. We define a language
${\cal L}$ in predicate logic containing a constant $\varepsilon$; for
each symbol $\gamma$ of the alphabet of ${\cal A}$, a unary function
symbol, also written $\gamma$; and for each state $P$ of ${\cal A}$
a unary predicate symbol, also written $P$.  A closed term in ${\cal L}$ 
has the form 
$\gamma_1 (\gamma_2 ... (\gamma_n (\varepsilon)))$,
where $\gamma_1$, ..., $\gamma_n$ are function symbols.  
Such a term is called a {\em word}, written $w = \gamma_1 \gamma_2
... \gamma_n$.  A closed atomic proposition has the form $P(w)$, 
where $P$ is a state and $w$ a word. 
We build an inference system that consists of, 
for each transition rule $P \xrightarrow{~\gamma~} Q$ of 
$\cal A$, the introduction rule
$${\small \irule{Q(x)}{P(\gamma(x))}{}}$$
and, for each final state $F$ of $\cal A$, the introduction rule 
$${\small \irule{}{F(\varepsilon)}{}}$$
It is routine to check that a word $w$ is recognized by the automaton
${\cal A}$ in a state $I$ if and only if the proposition $I(w)$ has
a proof in the corresponding system.

\section{From cut-elimination to automata}
\label{secdecidable}

In this section, we present two cut-elimination theorems, that permit
to completely eliminate the non-introduction rules and prove, this
way, the decidability of Finite Domain Logic and of Alternating
Pushdown Systems, respectively.

\subsection{Finite Domain Logic}
\label{secfdl}

\begin{figure}[t]
{\scriptsize
\noindent\framebox{\parbox{\textwidth
}{
{
~~~~~~~~~~~~~~$\begin{array}{lr}
\irule{}
      {\Gamma, A \vdash A}
      {\mbox{axiom}}
~~~~~~~~~~~~~~~~~~~~~~~~~~~~~~~~~~~~~~~~~~~~~~~~
\\
\irule{}
      {\Gamma \vdash L}
      {\mbox{atom if $L \in {\cal P}$}}
\\
\irule{}
      {\Gamma \vdash \top}
      {\mbox{$\top$-intro}} 
& 
\\
& 
\irule{\Gamma \vdash \bot}
      {\Gamma \vdash A}
      {\mbox{$\bot$-elim}} 
\\
\\
\irule{\Gamma \vdash A~~~\Gamma \vdash B}
      {\Gamma \vdash A \wedge B}
      {\mbox{$\wedge$-intro}} 
& 
\irule{\Gamma \vdash A \wedge B}
      {\Gamma \vdash A}
      {\mbox{$\wedge$-elim}} 
\\
& 
\irule{\Gamma \vdash A \wedge B}
      {\Gamma \vdash B}
      {\mbox{$\wedge$-elim}} 
\\
\\
\irule{\Gamma \vdash A}
      {\Gamma \vdash A \vee B}
      {\mbox{$\vee$-intro}}
&
\irule{\Gamma \vdash A \vee B~~~\Gamma, A \vdash C~~~\Gamma, B \vdash C}
      {\Gamma \vdash C}
      {\mbox{$\vee$-elim}}
\\
\irule{\Gamma \vdash B}
      {\Gamma \vdash A \vee B}
      {\mbox{$\vee$-intro}}
\\
\\
\irule{\Gamma \vdash (c_1/x)A~~~...~~~\Gamma \vdash (c_n/x)A}
      {\Gamma \vdash \fa x A}
      {\mbox{$\fa$-intro}}
& 
\irule{\Gamma \vdash \fa x A}
      {\Gamma \vdash (c_i/x)A}
      {\mbox{$\fa$-elim}} 
\\
\\
\irule{\Gamma \vdash (c_i/x)A}
      {\Gamma \vdash \ex x~A}
      {\mbox{$\ex$-intro}}
&
\irule{\Gamma \vdash \ex x~A~~~\Gamma, (c_1/x)A \vdash C~~~...~~~
\Gamma, (c_n/x)A \vdash C}
      {\Gamma \vdash C}
      {\mbox{$\ex$-elim}}
\end{array}$
\begin{center}
\caption{Finite Domain Logic \label{deffdl}}
\end{center}
}}}
}
\end{figure}

We begin with a toy example, {\em Finite Domain Logic}, 
motivated by its simplicity: we can prove a cut-elimination
theorem, showing the system is equivalent to the automaton
obtained by dropping its non-introduction rules.

Finite Domain Logic is a version of Natural Deduction tailored to
prove the propositions that are valid in a given finite model $\cal M$. The
differences with the usual Natural Deduction are the following: 
a proposition of the form $A \Rightarrow B$ is just an abbreviation for $\neg
A \vee B$ and negation has been pushed to atomic propositions
using de Morgan's laws;
the $\fa$-intro  and the $\ex$-elim rules are replaced by enumeration rules, 
and an {\em atom} rule is added to prove
closed atomic propositions and their negations 
valid in the underlying model.

If the model ${\cal M}$ is formed with a domain $\{a_1, ..., a_n\}$
and relations $R_1$, ..., $R_m$ over this domain, we consider the
language containing constants $c_1, ..., c_n$ for the elements $a_1,
..., a_n$ and predicate symbols $P_1, ..., P_m$ for the relations
$R_1, ..., R_m$.  The {\em Finite Domain Logic} of the model ${\cal
  M}$ is defined by the inference system of Figure \ref{deffdl}, where
the set ${\cal P}$ contains, for each atomic proposition $P_i(c_{j_1},
..., c_{j_k})$, either the proposition $P_i(c_{j_1}, ..., c_{j_k})$ if
$\langle a_{j_1}, ..., a_{j_k} \rangle$ is in $R_i$, or the
proposition $\neg P_i(c_{j_1}, ..., c_{j_k})$, otherwise.

In this system, the introduction rules are those presented in the
first column: the axiom rule, the atom rule, and the rules
$\top$-intro, $\wedge$-intro, $\vee$-intro, $\fa$-intro, and
$\ex$-intro.  The non-introduction rules are those presented in the
second column.  Each rule has one major premise: the leftmost one.  A
cut is as in Definition \ref{defgeneralcut}.

\begin{theorem}[Soundness, Completeness, and Cut-elimination]
Let $B$ be a closed proposition, the following are equivalent:
(1.) the proposition $B$ has a proof,
(2.) the proposition $B$ is valid in ${\cal M}$, 
(3.) the proposition 
$B$ has a cut-free proof, that is a proof formed with introduction 
rules only.
\end{theorem}

\VL{

\begin{proof}
To prove that {\em 1.} implies {\em 2.} (soundness), we actually prove, 
by induction over proof structure, a more general result: if $A_1$, ...,
$A_n$, $B$ are closed propositions and the sequent $A_1, ..., A_n
\vdash B$ has a proof, then the proposition $A_1 \Rightarrow
... \Rightarrow A_n \Rightarrow B$ is valid in ${\cal M}$.
To prove that {\em 2.} implies {\em 3.} (strengthened completeness), 
we proceed by induction over the structure of $B$. 
That {\em 3.} implies {\em 1.} is trivial, as cut-free proofs are 
proofs.
\end{proof}

}

Therefore, provability in Finite Domain Logic is decidable, as the provable
propositions are recognized by the automaton obtained by dropping
the non-introduction rules.  
Since the introduction rules preserve context emptiness,
the contexts can be ignored and the axiom rule can be dropped.  
This automaton could also be expressed in a more familiar way
with the transition rules
$${\small \begin{array}{rclrcl}
L & \hookrightarrow & \varnothing~\mbox{if $L \in {\cal P}$}~~~&
A \vee B & \hookrightarrow & \{A\}\\
\top & \hookrightarrow & \varnothing&
A \vee B & \hookrightarrow & \{B\}\\
A \wedge B & \hookrightarrow & \{A, B\}&
\fa x~A & \hookrightarrow & \{(c_1/x)A, ..., (c_n/x)A\}\\
&&&
\ex x~A & \hookrightarrow & \{(c_i/x)A\}~\mbox{for each $c_i$}
\end{array}}$$
\subsection{Alternating Pushdown Systems}
\label{aps}

\begin{figure}[t]
{\scriptsize
\noindent\framebox{\parbox{\textwidth
}{
{
~~~~~~~~~~~~~~$\begin{array}{ll}
\\
\irule{P_1(x)~...~P_n(x)}
      {Q(\gamma(x))}
      {\mbox{intro}~~~n \geq 0}
~~~~~~~~~~~~~~~~~~~~~~~~~~~~~~~~~~~~~~~~~~~~~~~~
&
\irule{P_1(\gamma(x))~P_2(x)~...~P_n(x)}
      {Q(x)}
      {\mbox{elim}~~~n \geq 1}
\\
\\
\irule{}
      {Q(\varepsilon)}
      {\mbox{intro}}
&
\irule{P_1(x)~...~P_n(x)} 
      {Q(x)} 
      {\mbox{neutral}~~~n \geq 0}
\end{array}$
\caption{Alternating Pushdown Systems \label{pushdown}}}}}}
\end{figure}

The second example, {\em Alternating Pushdown Systems}, is still
decidable \cite{BEM}, but a little bit more complex. Indeed
these systems, in general, 
need to be saturated---that is extended with
derivable rules---in order to enjoy cut-elimination.

Consider a language ${\cal L}$ containing a finite number of unary
predicate symbols, a finite number of unary function symbols, and a
constant $\varepsilon$.  An {\em Alternating Pushdown System} is an
inference system whose rules are like those presented in Figure
\ref{pushdown}. The rules in the first column are introduction rules
and those in the second column, the elimination and neutral rules, are
not. Elimination rules have one major premise, the leftmost one, and
all the premises of a neutral rule are major.  A cut is as in 
Definition \ref{defgeneralcut}.

Not all Alternating Pushdown Systems enjoy the cut-elimination
property.  However, every Alternating Pushdown System has an extension with
derivable rules that enjoys this property:
each time we have a cut of the form
$${\small 
  \irule{\irule{\irule{\rho^1_1}{s^1_1}{}~~...~~\irule{\rho^1_{m_1}}{s^1_{m_1}}{}}
               {s^1}
               {\mbox{intro}}
         ~~~~~~~~~~~~~~~~~~...~~~~~~~~~~
         \irule{\irule{\rho^k_1}{s^k_1}{}~~...~~\irule{\rho^k_{m_k}}{s^k_{m_k}}{}}
               {s^k}
               {\mbox{intro}}
          ~~~~~~~~~~~~~~~~~
         \irule{\pi_{k+1}}{s^{k+1}}{}
         ~~...~~
         \irule{\pi_n}{s^n}{}
        }
        {s}
        {\mbox{non-intro}}}$$
we add a derivable rule allowing to deduce directly $s$ from $s^1_1$,
..., $s^1_{m_1}$, ..., $s^k_1$, ..., $s^k_{m_k}$, $s^{k+1}$, ..., $s^n$.  
This leads to the following saturation algorithm
\cite{DJ,FruhwirthShapiroVardiYardeni,Goubault}.

\begin{definition}[Saturation]
Given an  Alternating Pushdown System,
\begin{itemize}
\item
if  it contains an introduction rule
$${\small \irule{P_1(x)~~...~~P_m(x)}{Q_1(\gamma(x))}{\mbox{intro}}}$$
and an elimination rule
$${\small \irule{Q_1(\gamma(x))~Q_2(x)~~...~~Q_n(x)}{R(x)}{\mbox{elim}}}$$
then we add the neutral rule 

$${\small \irule{P_1(x)~...~P_m(x)~Q_2(x)~~...~~Q_n(x)}{R(x)}{\mbox{neutral}}}$$

\item
if it contains introduction rules 
$${\small 
\begin{array}{ccc}
\irule{P^1_1(x)~~...~~P^1_{m_1}(x)}{Q_1(\gamma(x))}{\mbox{intro}}
&
~~~~~~~~~~~~~~~~~...~~~~~~~~~~
&
\irule{P^n_1(x)~~...~~P^n_{m_n}(x)}{Q_n(\gamma(x))}{\mbox{intro}}
\end{array}}$$
and a neutral rule 
$${\small \irule{Q_1(x)~~...~~Q_n(x)}{R(x)}{\mbox{neutral}}}$$
then we add the introduction rule 
$${\small 
  \irule{P^1_1(x)~...~P^1_{m_1}(x)~~...~~P^n_1(x)~...~P^n_{m_n}(x)}
        {R(\gamma(x))}
        {\mbox{intro}}}$$

\item
if it contains introduction rules
$${\small \begin{array}{ccc}
\irule{}{Q_1(\varepsilon)}{\mbox{intro}}
&
~~~~~~~~~~~~~~~~~...~~~~~~~~~~
&
\irule{}{Q_n(\varepsilon)}{\mbox{intro}}
\end{array}}$$
and a neutral rule 
$${\small \irule{Q_1(x)~~...~~Q_n(x)}{R(x)}{\mbox{neutral}}}$$
then we add the introduction rule 
$${\small \irule{}{R(\varepsilon)}{\mbox{intro}}}$$
\end{itemize}
As there is only a finite number of possible rules, this procedure
terminates.
\end{definition}

It is then routine to check that if a closed proposition has a proof 
in a saturated system, it has a cut-free proof \cite{DJ}, leading
to the following result.

\begin{theorem}[Decidability]
Provability of a closed proposition in an Alternating Pushdown System
is decidable.
\end{theorem}

\VL{

\begin{proof}
Consider an Alternating Pushdown System $S$, the system 
$S'$ obtained by saturating $S$, and the automaton
$S''$ obtained by dropping all the non-introduction rules in $S'$. 
These three systems prove the same closed propositions and 
provability in $S''$ is decidable.
\end{proof}

}

\begin{example}
Consider the Alternating Pushdown System $S$
$${\small \begin{array}{llll}
\irule{Q(x)}{P(a x)}{\mbox{\bf i1}}
~~~~~~~~~~~~~~
&
\irule{T(x)}{P(b x)}{\mbox{\bf i2}}
~~~~~~~~~~~~~~
&
\irule{T(x)}{R(a x)}{\mbox{\bf i3}}
~~~~~~~~~~~~~~
&
\irule{}{R(b x)}{\mbox{\bf i4}}
\\
\\
\irule{P(x)~R(x)}{Q(x)}{\mbox{\bf n1}}
~~~~~~~~~~~~~~
&
\irule{}{T(x)}{\mbox{\bf n2}}
&
\irule{P(ax)}{S(x)}{\mbox{\bf e1}}
\end{array}}$$
The system $S'$ obtained by saturating the system $S$ 
contains the rules of the system $S$ and the following rules 
$${\small 
\begin{array}{llll}
\irule{Q(x)}{S(x)}{\mbox{\bf n3}}
~~~~~~~~~~~~~~
& 
\irule{}{T(\varepsilon)}{\mbox{\bf i5}}
~~~~~~~~~~~~~~
&

\irule{}{T(a x)}{\mbox{\bf i6}}
~~~~~~~~~~~~~~
&
\irule{Q(x)~T(x)}{Q(ax)}{\mbox{\bf i7}}
\\
\\
\irule{Q(x)~T(x)}{S(ax)}{\mbox{\bf i8}}
~~~~~~~~~~~~~~
&
\irule{}{T(b x)}{\mbox{\bf i9}}
~~~~~~~~~~~~~~
&
\irule{T(x)}{Q(bx)}{\mbox{\bf i10}}
~~~~~~~~~~~~~~
&
\irule{T(x)}{S(bx)}{\mbox{\bf i11}}
\end{array}}$$
The automaton $S''$ contain the rules 
{\bf i1}, {\bf i2}, {\bf i3}, {\bf i4}, 
{\bf i5}, {\bf i6}, {\bf i7}, {\bf i8}, 
{\bf i9}, {\bf i10}, {\bf i11}. 

The proof in the system $S$
$${\small 
  \irule{\irule{\irule{\irule{\irule{\irule{\irule{}
                                                  {T(\varepsilon)}
                                                  {\mbox{\bf n2}}
                                           }
                                           {P(b)}
                                           {\mbox{\bf i2}}
                                     ~~~~~~~~~~~~~~~
                                     \irule{}
                                           {R(b)}
                                           {\mbox{\bf i4}}
                                    }
                                    {Q(b)}
                                    {\mbox{\bf n1}}
                             }
                             {P(ab)}
                             {\mbox{\bf i1}}
                      ~~~~~~~~~~~~~~~~~~~~~~~
                      \irule{\irule{}{T(b)}{\mbox{\bf n2}}}
                            {R(ab)}
                            {\mbox{\bf i3}}
                      }
                      {Q(ab)}
                      {\mbox{\bf n1}}
               }
               {P(aab)}
               {\mbox{\bf i1}}
         }
         {S(ab)}
         {\mbox{\bf e1}}}$$
reduces to the cut-free proof in the system $S''$
$${\small 
  \irule{\irule{\irule{}{T(\varepsilon)}{\mbox{\bf i5}}}
               {Q(b)}
               {\mbox{\bf i10}}
        ~~~~~~~~~~
        \irule{}{T(b)}{\mbox{\bf i9}}
        }
        {S(ab)}
        {\mbox{\bf i8}}}$$
\end{example}

\section{Partial results for undecidable systems}
\label{secpartial}

In this section, we focus on Constructive Predicate Logic, leaving the
case of Classical Predicate Logic for future work.  We start with
Natural Deduction \cite{Prawitz}. As provability in Predicate Logic is
undecidable, we cannot expect to transform Natural Deduction into an
automaton. But, as we shall see, saturation permits to transform first
Natural Deduction into a Gentzen style Sequent Calculus \cite{Kleene},
then the latter into a Kleene style Sequent Calculus \cite{Kleene},
and then the latter into a Vorob'ev-Hudelmaier-Dyckhoff-Negri style
Sequent Calculus \cite{Vorobev,Hudelmaier,Dyckhoff,DyckhoffNegri}.
Each time, a larger fragment of Constructive Predicate Logic is proved 
decidable.

Note that each transformation proceeds in the same way: 
first, we identify some general cuts.
Then, like in the saturation procedure of Section \ref{aps},
we add some admissible rules to eliminate these cuts. 
Finally, we prove a cut-elimination theorem showing that some
non-introduction rules can be dropped.

\subsection{Natural Deduction}
\label{secND}

\begin{figure}[t]
{\scriptsize
\noindent\framebox{\parbox{\textwidth
}{
{
~~~~~$\begin{array}{lr}
\irule{}
      {\Gamma, A \vdash A}
      {\mbox{axiom}}
~~~~~~~~~~~~~~~~~~~~~~~~~~~~~~~~~~~~~~~~~~~~~~
\\
&
\\
\irule{}
      {\Gamma \vdash \top}
      {\mbox{$\top$-intro}}
&
\\
&
\irule{\Gamma \vdash \bot}
      {\Gamma \vdash A}
      {\mbox{$\bot$-elim}}
\\
\\
\irule{\Gamma \vdash A~~~\Gamma \vdash B}
      {\Gamma \vdash A \wedge B}
      {\mbox{$\wedge$-intro}}
&
\irule{\Gamma \vdash A \wedge B}
      {\Gamma \vdash A}
      {\mbox{$\wedge$-elim}}
\\
&
\irule{\Gamma \vdash A \wedge B}
      {\Gamma \vdash B}
      {\mbox{$\wedge$-elim}}
\\
\irule{\Gamma \vdash A}
      {\Gamma \vdash A \vee B}
      {\mbox{$\vee$-intro}}
\\
\\
\irule{\Gamma \vdash B}
      {\Gamma \vdash A \vee B}
      {\mbox{$\vee$-intro}}
&
\irule{\Gamma \vdash A \vee B~~~\Gamma, A \vdash C~~~\Gamma, B \vdash C}
      {\Gamma \vdash C}
      {\mbox{$\vee$-elim}}
\\
\\
\irule{\Gamma, A \vdash B}
      {\Gamma \vdash A \Rightarrow B}
      {\mbox{$\Rightarrow$-intro}}
&
\irule{\Gamma \vdash A \Rightarrow B~~~\Gamma \vdash A}
      {\Gamma \vdash B}
      {\mbox{$\Rightarrow$-elim}}
\\
\\
\irule{\Gamma \vdash A}
      {\Gamma \vdash \fa x A}
      {\mbox{$\fa$-intro if $x$ not free in $\Gamma$}}
&
\irule{\Gamma \vdash \fa x A}
      {\Gamma \vdash (t/x)A}
      {\mbox{$\fa$-elim}}
\\
\\
\irule{\Gamma \vdash (t/x)A}
      {\Gamma \vdash \ex x~A}
      {\mbox{$\ex$-intro}}
&
\irule{\Gamma \vdash \ex x~A~~~\Gamma, A \vdash B}
      {\Gamma \vdash B}
      {\mbox{$\ex$-elim if $x$ not free in $\Gamma, B$}}\\
\end{array}$
\begin{center}
\caption{Constructive Natural Deduction\label{natded}}
\end{center}
}}}
}
\end{figure}

In Natural Deduction (Figure \ref{natded}), the introduction rules are
those presented in the first column, they are the axiom rule and the
rules $\top$-intro, $\wedge$-intro, $\vee$-intro, $\Rightarrow$-intro,
$\fa$-intro, and $\ex$-intro.  The non-introduction rules are those
presented in the second column, each of them has one major premise:
the leftmost one.

Natural Deduction has a specific notion of cut: a proof ending
with a $\wedge$-elim, $\vee$-elim, $\Rightarrow$-elim, $\fa$-elim,
$\ex$-elim rule, whose major premise is proved with a proof ending with
a $\wedge$-intro, $\vee$-intro, $\Rightarrow$-intro, $\fa$-intro,
$\ex$-intro rule, respectively.  
The only difference between this specific notion of cut and the general one 
(Definition \ref{defgeneralcut}) is that the general notion has one more 
form of cut: a proof built with an elimination rule whose major premise is 
proved with the axiom rule. For instance
$${\small 
  \irule{\irule{}
               {P \wedge Q \vdash P \wedge Q}
               {\mbox{axiom}}
        }
        {P \wedge Q \vdash P}
        {\mbox{$\wedge$-elim}}}$$
So proofs free of specific cuts can still contain general cuts of this form.

Saturating the system, like in Section \ref{aps}, to eliminate the 
specific cuts, would add \VL{the }derivable rules \VC{such as}
$${\small 
  \irule{\Gamma \vdash A~~~~~\Gamma \vdash B}
        {\Gamma \vdash A}
        {R_\wedge}}$$
\VL{
$${\small 
  \irule{\Gamma \vdash A~~~~~\Gamma \vdash B}
        {\Gamma \vdash B}
        {R_\wedge}}$$
$${\small 
  \irule{\Gamma \vdash A~~~\Gamma, A \vdash C~~~\Gamma, B \vdash C}
        {\Gamma \vdash C}
        {R_\vee}}$$
$${\small 
  \irule{\Gamma \vdash B~~~\Gamma, A \vdash C~~~\Gamma, B \vdash C}
        {\Gamma \vdash C}
        {R_\vee}}$$
$${\small
  \irule{\Gamma, A \vdash B~~~\Gamma \vdash A}
        {\Gamma \vdash B}
        {R_\Rightarrow}}$$
$${\small 
  \irule{\Gamma \vdash A}
        {\Gamma \vdash (t/x)A}
        {\mbox{$R_\fa$ ~ $x$ not free in $\Gamma$}}}$$
$${\small 
  \irule{\Gamma \vdash (t/x)A~~~~\Gamma, A \vdash C}
        {\Gamma \vdash C}
        {\mbox{$R_\ex$ ~ $x$ not free in $\Gamma$ and $C$}}}$$
}
But they are not needed, as they are admissible in cut-free Natural 
Deduction.

\VL{
For
instance, instead of building the proof
$${\small 
  \irule{\irule{\pi_1}{\Gamma \vdash A}{}
         ~~~~~
         \irule{\pi_2}{\Gamma \vdash B}{}
        }
        {\Gamma \vdash A}
        {R_\wedge}}$$
we can just build the proof $\pi_1$. 
The rule $R_\fa$ is just a substitution rule and instead of 
building the proof 
$${\small  
  \irule{\irule{\pi}
               {\Gamma \vdash A}
               {}
        }
        {\Gamma \vdash (t/x)A}
        {R_\fa}}$$
we can just build the proof $(t/x)\pi$, where the term $t$ 
is substituted for $x$ everywhere in $\pi$. 
The rule $R_\Rightarrow$ is just a substitution rule and instead of 
building the proof 
$${\small 
  \irule{\irule{\pi_1}{\Gamma, A \vdash B}{}
         ~~~
         \irule{\pi_2}{\Gamma \vdash A}{}
        }
        {\Gamma \vdash B}
        {R_\Rightarrow}}$$
we can just build the proof $(\pi_2/\langle A \rangle)\pi_1$, where the
proof $\pi_2$ is substituted for the axiom rule on $A$ everywhere in
$\pi_1$, etc.}

\VC{The admissibility of some rules 
however are based on a substitution of proofs, that may create new cuts
on smaller propositions, that need in turn to be eliminated. 
In other words, the termination of the specific cut-elimination algorithm 
needs to be proved \cite{Prawitz}.}

\VL{
In the case of the rule $R_\Rightarrow$ however, the substitution $(\pi_2/\langle A \rangle)\pi_1$
may create cuts on smaller propositions, so we need to transform this proof 
further into a cut-free proof, that is to prove the termination 
of the cut-elimination algorithm \cite{Prawitz}.}

As general cuts with an axiom rule are not eliminated, this partial
cut-elimination theorem is not sufficient to eliminate all
elimination rules and to prove the decidability of Constructive
Natural Deduction, but it yields a weaker result: a (specific-)cut-free proof
ends with introduction rules, as long as the context of the proved
sequent contains atomic propositions only.  To formalize this result,
\VL{let us}\VC{we} introduce a modality $[~]$ and define a translation
that freezes the non atomic left-hand parts of 
implications\VL{.}\VC{, 
$f(A \Rightarrow B) = [A] \Rightarrow f(B)$, if $A$ is not
  atomic, and 
$f(A \Rightarrow B) = A \Rightarrow f(B)$, if $A$ is atomic, 
$f(A \wedge B) = f(A) \wedge f(B)$, etc., and the converse function 
$u$ is defined in a trivial way.}

\VL{
\begin{definition}[Freeze, unfreeze]
The freeze function $f$ is defined as follows
\begin{itemize}
\item $f(A \Rightarrow B) = [A] \Rightarrow f(B)$, if $A$ is not
  atomic,
\item $f(P \Rightarrow B) = P \Rightarrow f(B)$, if $P$ is atomic,
\item $f(P) = P$, if $P$ is atomic,
\item $f(\top) = \top$, $f(\bot) = \bot$,
\item $f(A \wedge B) = f(A) \wedge f(B)$, $f(A \vee B) = f(A) \vee
  f(B)$,
\item $f(\fa x A) = \fa x f(A)$, $f(\ex x~A) = \ex x~f(A)$.
\end{itemize}

The converse unfreeze function $u$ is defined as follows

\begin{itemize}
\item $u([A]) = u(A)$,
\item $u(P) = P$, if $P$ is atomic,
\item $u(\top) = \top$, $u(\bot) = \bot$,
\item $u(A \Rightarrow B) = u(A) \Rightarrow u(B)$,
\item $u(A \wedge B) = u(A) \wedge u(B)$, $u(A \vee B) = u(A) \vee
  u(B)$,
\item $u(\fa x A) = \fa x u(A)$, $u(\ex x~A) = \ex x~u(A)$.
\end{itemize}

These definitions extend to sequences of propositions and to
sequents in a trivial way.
\end{definition}
}

\begin{definition}
Let ${\cal A}$ be the pseudo-automaton formed with the introduction rules
of Constructive Natural Deduction, including the axiom rule, plus the
introduction rule
$${\small \irule{} {\Gamma, [A] \vdash B} {\mbox{delay}}}$$
\end{definition}

\VL{

We want to prove that if a proposition $A$ has a (specific-)cut-free proof in
Constructive Natural Deduction, then the proposition $f(A)$ has a
proof in ${\cal A}$ and for each leaf $\Delta \vdash B$ proved with
the delay rule, the sequent $u(\Delta \vdash B)$ has a proof in
Constructive Natural Deduction.  To do so, we first need the following
lemma.

\begin{lemma}[Last rule]
\label{lemmalastrule}
If $\pi$ is a (specific-)cut-free proof of $\Gamma \vdash A$ in Constructive
Natural Deduction, where $\Gamma$ contains atomic propositions only,
then $\pi$ ends with an introduction rule.
\end{lemma}

\begin{proof}
By induction over proof structure. The proof $\pi$ has the form
$${\small 
  \irule{\irule{\pi_1}{\Gamma_1 \vdash B_1}{}
         ~~~~
         \irule{\pi_2}{\Gamma_2 \vdash B_2}{}
         ~~~~...~~~~
         \irule{\pi_p}{\Gamma_p \vdash B_p}{}\
        }
        {\Gamma \vdash A}
        {R}}$$
Suppose the rule $R$ is an elimination rule. Then $\Gamma_1 = \Gamma$
and, by induction hypothesis, the last rule of $\pi_1$ is an
introduction rule.  If it is an axiom rule, then the proposition $B_1$
is atomic, contradicting the assumption that $R$ is an elimination
rule. Thus, it is one of the rules $\top$-intro, $\wedge$-intro,
$\vee$-intro, $\fa$-intro, $\ex$-intro, or $\Rightarrow$-intro,
contradicting the fact that the proof is (specific-)cut-free.
\end{proof}

Note that, in this lemma, we do not need to consider commutative cuts,
because the context $\Gamma$ contains atomic propositions only.

}

\begin{theorem}
\label{mainND}
Let $\Gamma \vdash A$ be a sequent such that $\Gamma$ contains atomic
propositions only. If $\Gamma \vdash A$ has a (specific-)cut-free proof in
Constructive Natural Deduction, then $\Gamma \vdash f(A)$ has a proof
in the pseudo-automaton ${\cal A}$ and for each leaf $\Delta \vdash B$
proved with the delay rule, the sequent $u(\Delta \vdash B)$ has a
proof in Constructive Natural Deduction.
\end{theorem}

A first corollary of Theorem \ref{mainND} is the decidability of the
small fragment
$${\small A = P~|~\top~|~\bot~|~A \wedge A~|~A \vee A~|~P \Rightarrow A~|~\fa x~A~|~\ex x~A}$$
where the left-hand side of an implication is always atomic, that is
no connective or quantifier has a negative occurrence.  
\VC{As the pseudo-automaton obtained this way is not finitely
branching, we need, as well-known, to introduce meta-variables to 
prove this decidability result.}

\VL{
Indeed, if $A$ is a proposition in this fragment, then $f(A) = A$ and,
if $\vdash A$ has a proof, then it has a proof in the pseudo-automaton
${\cal A}$.  As ${\cal A}$ is a pseudo-automaton, but not a finitely
branching one, this is not sufficient to prove the decidability of
this fragment.  Nevertheless, we can prove this decidability by using
a well-known method: introducing meta-variables
$${\small 
  \irule{\Gamma \vdash (X/x)A}
        {\Gamma \vdash \ex x~A}
        {\mbox{$\ex$-intro}}}$$
and, in the axiom rule, replacing the comparison of propositions with
global unification.
}

A second corollary is that if $A$ is a proposition starting with
$n$ connectors or quantifiers different from $\Rightarrow$, then a
(specific-)cut-free proof of the sequent $\vdash A$ ends with $n+1$
successive introduction rules.  For $n = 0$, we obtain the
well-known last rule property of constructive (specific-)cut-free
proofs\VL{ (Lemma \ref{lemmalastrule})}. For a proposition $A$ of the
form $\fa x ~(B_1 \vee B_2)$, for instance, we obtain that a
(specific-)cut-free proof of $\vdash \fa x ~(B_1 \vee B_2)$ ends with
three introduction rules. Thus, it has the form
$${\small 
  \irule{\irule{\irule{\pi'}
                      {\vdash B_i}
                      {}
               }
               {\vdash B_1 \vee B_2}
               {\mbox{$\vee$-intro}}
        }
        {\vdash \fa x~(B_1 \vee B_2)}
        {\mbox{$\fa$-intro}}}$$
and $\pi'$ itself ends with an introduction rule.  As a consequence, 
if the proposition $\fa x~(B_1 \vee B_2)$ has a
proof, then either the proposition $B_1$ or the proposition $B_2$ has a
proof, thus the proposition $(\fa x~ B_1) \vee (\fa x ~B_2)$
has a proof. This commutation of the universal quantifier with the
disjunction is called a {\em shocking equality} \cite{Girard}.

\subsection{Eliminating elimination rules: Gentzen style Sequent Calculus}

To eliminate the general cuts of the form 
$${\small 
  \irule{\irule{}
               {A \wedge B \vdash A \wedge B}
               {\mbox{axiom}}
        }
        {A \wedge B \vdash A}
        {\mbox{$\wedge$-elim}}}$$
we could add an introduction rule of the form 
$${\small 
  \irule{}
        {A \wedge B \vdash A}
        {I}}$$
But, this saturation procedure would not terminate. 
\VL{For instance, the proof 
$${\small
  \irule{\irule{\irule{}
                      {(A \wedge B) \wedge C \vdash (A \wedge B) \wedge C}
                      {\mbox{axiom}}
               }
               {(A \wedge B) \wedge C \vdash A \wedge B}
               {\mbox{$\wedge$-elim}}
       }
       {(A \wedge B) \wedge C \vdash A}
       {\mbox{$\wedge$-elim}}}$$
would reduce to 
$${\small 
  \irule{\irule{}
               {(A \wedge B) \wedge C \vdash A \wedge B}
               {I}
       }
       {(A \wedge B) \wedge C \vdash A}
       {\mbox{$\wedge$-elim}}}$$
that contains a general cut formed with the rule $I$ and the
$\wedge$-elim rule. And, to eliminate this cut, we would need yet
another introduction rule
$${\small 
  \irule{}
        {(A \wedge B) \wedge C \vdash A}
        {I'}}$$
and so on.
}

A way to keep the number of rules finite is to add left introduction rules 
to decompose \VL{these complex propositions}\VC{the complex hypotheses}, 
before they are used by the axiom 
rule\VC{: the left rules of Sequent Calculus.}\VL{, for example
$${\small 
  \irule{\irule{\irule{}
                      {A, B, C \vdash A}
                      {\mbox{axiom}}
               }
               {A \wedge B, C \vdash A}
               {\mbox{$\wedge$-left}}
       }
       {(A \wedge B) \wedge C \vdash A}
       {\mbox{$\wedge$-left}}}$$
The introduction rules of Natural Deduction are then called 
{\em right introduction rules}.

}
However, this is still not sufficient to eliminate the elimination rules of
Constructive Natural Deduction.  For instance, the sequent $\fa x~
(P(x) \wedge (P(f(x)) \Rightarrow Q)) \vdash Q$ has a proof using 
elimination rules
$$\hspace*{-1.5cm} {\small 
  \irule{\irule{\irule{\irule{}
                             {\Gamma \vdash \fa x~ (P(x) \wedge (P(f(x)) \Rightarrow Q))}
                             {\mbox{axiom}}
                      }
                      {\Gamma \vdash P(c) \wedge (P(f(c)) \Rightarrow Q)}
                      {\mbox{$\fa$-elim}}
               }
               {\Gamma \vdash P(f(c)) \Rightarrow Q}
               {\mbox{$\wedge$-elim}}
          ~~~~~~~~~~~~~~~~~~~~~~~~~~~~~~~~~~~~~
         \irule{\irule{\irule{}
                             {\Gamma \vdash \fa x ~ (P(x) \wedge (P(f(x)) \Rightarrow Q))}
                             {\mbox{axiom}}
                      }
                      {\Gamma \vdash P(f(c)) \wedge (P(f(f(c))) \Rightarrow Q)}
                      {\mbox{$\fa$-elim}}
               }
               {\Gamma \vdash P(f(c))}
               {\mbox{$\wedge$-elim}}
        }
        {\Gamma \vdash Q}
        {\mbox{$\Rightarrow$-elim}}}$$
where $\Gamma = \fa x~(P(x) \wedge (P(f(x)) \Rightarrow Q))$, but none
using introduction rules only\VL{, as from the sequent
$\fa x~(P(x) \wedge (P(f(x)) \Rightarrow Q)) \vdash Q$, the only
introduction rule that can be applied is the $\fa$-left rule and, for
every term $t$, the sequent $P(t) \wedge (P(f(t)) \Rightarrow Q)
\vdash Q$ is not provable}.

\begin{figure}[t]
{\scriptsize
\noindent\framebox{\parbox{\textwidth
}{
{
~~~~~~~~~$
\begin{array}{lr}
\irule{}{\Gamma, P \vdash P}{\mbox{axiom $P$ atomic}}
~~~~~~~~~~~~~~~~~~~~~~~~~~~~~~~~~~~~~~~~~~
&
\irule{\Gamma, A, A \vdash G}
      {\Gamma, A \vdash G}
      {\mbox{contraction}}
\\
\irule{}
      {\Gamma \vdash \top}
      {\mbox{$\top$-right}}\\
\irule{}
      {\Gamma, \bot \vdash G}
      {\mbox{$\bot$-left}}\\
\irule{\Gamma, A, B \vdash G}
      {\Gamma, A \wedge B \vdash  G}
      {\mbox{$\wedge$-left}}\\
\irule{\Gamma \vdash A~~~\Gamma \vdash B}
      {\Gamma \vdash A \wedge B}
      {\mbox{$\wedge$-right}}\\
\irule{\Gamma, A \vdash G ~~~ \Gamma, B \vdash G}
      {\Gamma, A \vee B \vdash G}
      {\mbox{$\vee$-left}}\\
\irule{\Gamma \vdash A}
      {\Gamma \vdash A \vee B}
      {\mbox{$\vee$-right}}\\
\irule{\Gamma \vdash B}
      {\Gamma \vdash A \vee B}
      {\mbox{$\vee$-right}}\\
\irule{\Gamma \vdash A ~~~ \Gamma, B \vdash G}
      {\Gamma, A \Rightarrow B \vdash  G}
      {\mbox{$\Rightarrow$-left}}\\
\irule{\Gamma, A \vdash B}
      {\Gamma \vdash  A \Rightarrow B}
      {\mbox{$\Rightarrow$-right}}\\
\irule{\Gamma, (t/x)A \vdash G}
      {\Gamma, \fa x A \vdash G}
      {\mbox{$\fa$-left}}\\
\irule{\Gamma \vdash A}
      {\Gamma \vdash \fa x A}
      {\mbox{$\fa$-right if $x$ not free in $\Gamma$}}\\
\irule{\Gamma, A \vdash G}
      {\Gamma, \ex x~A \vdash G}
      {\mbox{$\ex$-left if $x$ not free in $\Gamma, G$}}\\
\irule{\Gamma \vdash (t/x)A}
      {\Gamma \vdash \ex x~A}
      {\mbox{$\ex$-right}}
\end{array}$
\begin{center}
\caption{Gentzen style Sequent Calculus: the system $\cal G$
\label{constrsequent}}
\end{center}}}}
}
\end{figure}

So, we need to add a contraction rule, to use an hypothesis several times
$${\small \irule{\Gamma, A, A \vdash G} {\Gamma, A \vdash G}{\mbox{contraction}}}$$
To prove that the elimination rules of Natural Deduction can now be eliminated, we prove\VC{, using Gentzen's theorem \cite{GirardLafontTaylor},}
that 
they are admissible in the 
system $\cal G$
(Figure \ref{constrsequent}), the Gentzen style Sequent Calculus, obtained 
by dropping the elimination rules of Constructive Natural Deduction. In 
this system, all the rules are 
introduction rules, except the contraction rule.
\VL{Note that the axiom 
rule is restricted to atomic propositions, as it is well-known that this 
restriction is immaterial.}
\VL{In order to prove that these elimination rules are admissible in the
system $\cal G$, we first need the following lemma (see, for instance,
\cite{GirardLafontTaylor}).

\begin{theorem}[Gentzen]
\label{Gentzen}
In the system $\cal G$,
the cut rule is admissible, that is,
if $\Gamma \vdash C$ and $\Gamma, C \vdash A$ both have proofs, then so does $\Gamma \vdash A$.
\end{theorem}

Then, the admissibility of the elimination rules is a simple
corollary.  From Theorem \ref{Gentzen}, if $\Gamma \vdash A
\wedge B$ and $\Gamma, A \wedge B \vdash A$ both have proofs,  then so does
$\Gamma \vdash A$.  But as the sequent $\Gamma, A \wedge B \vdash A$
has the trivial proof
$${\small 
  \irule{\irule{}
               {\Gamma, A, B \vdash A}
               {\mbox{axiom}}
        }
        {\Gamma, A \wedge B \vdash A}
        {\mbox{$\wedge$-left}}}$$
we obtain that if $\Gamma \vdash A \wedge B$ has a proof, then so does
$\Gamma \vdash A$, which is the admissibility of the
$\wedge$-elim rule.}
The system $\cal G$ does not allow to prove the decidability of any larger
fragment of Constructive Predicate Logic, but it is the basis of the
two systems presented in the Sections \ref{seckleene} and
\ref{secdychkoff}.

\subsection{Eliminating the contraction rule: Kleene style Sequent Calculus} 
\label{seckleene}

In the system $\cal G$, the proof 
$${\small 
  \irule{\irule{\irule{\rho}{\Gamma, \fa x ~A, (t/x)A \vdash B}{}}
               {\Gamma, \fa x ~ A, \fa x~ A \vdash B}
               {\mbox{$\fa$-left}}
        }
        {\Gamma, \fa x ~ A \vdash B}
        {\mbox{contraction}}}$$
is a general cut and we may replace it by the application of the
derivable rule
$${\small 
  \irule{\irule{\rho}{\Gamma, \fa x ~ A, (t/x)A \vdash B}{}}
        {\Gamma, \fa x ~ A \vdash B}
        {\mbox{contr-$\fa$-left}}}$$
which is a rule {\em \`a la} Kleene.
\VC{The other general cuts yields similar derivable rules.}\VL{
The other general cuts, where the last rule is a contraction and the
rule above is an introduction applied to the contracted proposition,
yield the derivable rules
$${\small 
  \irule{}
        {\Gamma, \bot \vdash G}
        {}}$$
$${\small 
  \irule{\Gamma, A \wedge B, A, B \vdash G}
        {\Gamma, A \wedge B \vdash  G}
        {}}$$
$${\small 
  \irule{\Gamma, A \vee B, A \vdash G~~~\Gamma, A \vee B, B \vdash G}
        {\Gamma, A \vee B \vdash G}
        {}}$$
$${\small 
  \irule{\Gamma, A \Rightarrow B \vdash A ~~~ \Gamma, A \Rightarrow B, B \vdash G}
        {\Gamma, A \Rightarrow B \vdash  G}
        {}}$$
$${\small 
  \irule{\Gamma, \ex x~A, A \vdash G}
        {\Gamma, \ex x~A \vdash G}
        {\mbox{if $x$ not free in $\Gamma, G$}}}$$
} 
But, as noticed by Kleene, the derivable rules for the contradiction, the conjunction, the
disjunction and the existential quantifier can be dropped, while that for
the implication can be simplified to
$${\small \irule{\Gamma, A \Rightarrow B \vdash A ~~~ \Gamma, B \vdash G}
        {\Gamma, A \Rightarrow B \vdash  G}
        {\mbox{contr-$\Rightarrow$-left}}}$$
The rules $\Rightarrow$-left and
$\fa$-left of the system $\cal G$, that are subsumed by the rules
contr-$\Rightarrow$-left and contr-$\fa$-left, can be also dropped.
There are also other general cuts, where the last rule is a
contraction and the rule above is an introduction applied to another
proposition, but these cuts can be eliminated without introducing any 
extra rule. In other words, after applying the contraction rule, 
we can focus on the contracted proposition \cite{DyckhoffLengrand}.

\begin{figure}[t]
{\scriptsize
\noindent\framebox{\parbox{\textwidth
}{
{
~~~~~~~~~$
\begin{array}{lr}
\irule{}
      {\Gamma, P \vdash P}
      {\mbox{axiom $P$ atomic}}
~~~~~~~~~~~~~~~~~~~~~~~~~~~~~~~~~~~~~~~~~~~~~~~~~\\
\irule{}
      {\Gamma \vdash \top}
      {\mbox{$\top$-right}}\\
\irule{}
      {\Gamma, \bot \vdash G}
      {\mbox{$\bot$-left}}\\
\irule{\Gamma, A, B \vdash G}
      {\Gamma, A \wedge B \vdash  G}
      {\mbox{$\wedge$-left}}\\
\irule{\Gamma \vdash A~~~\Gamma \vdash B}
      {\Gamma \vdash A \wedge B}
      {\mbox{$\wedge$-right}}\\
\irule{\Gamma, A \vdash G ~~~ \Gamma, B \vdash G}
      {\Gamma, A \vee B \vdash G}
      {\mbox{$\vee$-left}}\\
\irule{\Gamma \vdash A}
      {\Gamma \vdash A \vee B}
      {\mbox{$\vee$-right}}\\
\irule{\Gamma \vdash B}
      {\Gamma \vdash A \vee B}
      {\mbox{$\vee$-right}}\\
\irule{\Gamma, A \vdash B}
      {\Gamma \vdash  A \Rightarrow B}
      {\mbox{$\Rightarrow$-right}}
&
\irule{\Gamma, A \Rightarrow B \vdash A~~~\Gamma, B \vdash G}
      {\Gamma, A \Rightarrow B \vdash  G}
      {\mbox{contr-$\Rightarrow$-left}}\\
\irule{\Gamma \vdash A}
      {\Gamma \vdash \fa x ~ A}
      {\mbox{$\fa$-right if $x$ not free in $\Gamma$}}
&
\irule{\Gamma, \fa x ~ A, (t/x)A \vdash G}
      {\Gamma, \fa x ~ A \vdash G}
      {\mbox{contr-$\fa$-left}}\\
\irule{\Gamma, A \vdash G}
      {\Gamma, \ex x~A \vdash G}
      {\mbox{$\ex$-left if $x$ not free in $\Gamma, G$}}\\
\irule{\Gamma \vdash (t/x)A}
      {\Gamma \vdash \ex x~A}
      {\mbox{$\ex$-right}}
\end{array}$
\begin{center}
\caption{Kleene style Sequent Calculus: the system ${\cal K}$\label{k}}
\end{center}}}}
}
\end{figure}

We get this way the system $\cal K$ (Figure \ref{k}).  In this system,
all rules are introduction rules, except the rules
contr-$\Rightarrow$-left and contr-$\forall$-left. The system $\cal K$
plus the contraction rule is obviously sound and complete with respect
to the system $\cal G$.  To prove that the contraction rule can be
eliminated from it, and hence the system $\cal K$ also is sound and
complete with respect to the system $\cal G$, we prove the admissibility of the
contraction rule in the system $\cal K$\VC{---see the long 
version of the paper for the full proof}.
The system $\cal K$ gives the decidability of a larger fragment of
Constructive Predicate Logic, where the implication and the universal
quantifier have no negative occurrences.

\subsection{Eliminating the contr-$\Rightarrow$-left rule: 
Vorob'ev-Hudelmaier-Dyckhoff-Negri style Sequent Calculus}
\label{secdychkoff}

In order to eliminate the contr-$\Rightarrow$-left rule, we consider
the general cuts where a sequent $\Gamma, A \Rightarrow B \vdash G$ is
proved with a contr-$\Rightarrow$-left rule whose major premise
$\Gamma, A \Rightarrow B \vdash A$ is proved with an introduction rule,
applied to the proposition $A$.  \VL{There are also other general
  cuts, where the last rule is a contr-$\Rightarrow$-left rule and the
  rule above is an introduction applied to a proposition different
  from $A$, but, as we shall see, we do not need to introduce any rule
  to eliminate these cuts. In other words, after applying the
  contr-$\Rightarrow$-left, we can focus on the proposition $A$
  \cite{DyckhoffLengrand}.}  This leads to consider the various cases
for $A$, that is hypotheses of the form $P \Rightarrow B$, $\top
\Rightarrow B$, $(C \wedge D) \Rightarrow B$, $(C \vee D) \Rightarrow
B$, $(C \Rightarrow D) \Rightarrow B$, $(\fa x~C) \Rightarrow B$, and
$(\ex x~C) \Rightarrow B$. The case $A = P$, atomic, needs to be
considered because the premise $\Gamma, A \Rightarrow B \vdash A$ may
be proved with the axiom rule, but the case $\bot \Rightarrow B$ does
not, because there is no right rule for the symbol $\bot$. This
enumeration of the various shapes of $A$ is the base of the sequent
calculi in the style of Vorob'ev, Hudelmaier, Dyckhoff, and Negri
\cite{Vorobev,Hudelmaier,Dyckhoff,DyckhoffNegri}.

We obtain this way several types of general cuts that can be eliminated 
\VC{by introducing derivable rules.}\VL{with the following derivable rules
$${\small \irule{\Gamma, P, B \vdash G}
        {\Gamma, P, P \Rightarrow B \vdash G}
        {\mbox{$\Rightarrow$-left$_{\mbox{\scriptsize \em axiom}}$}}}$$
$${\small \irule{\Gamma, B \vdash G}
        {\Gamma, \top \Rightarrow B \vdash G}
        {\mbox{$\Rightarrow$-left$_\top$}}}$$
$${\small \irule{\Gamma, (C \wedge D) \Rightarrow B \vdash C
         ~~~~~
         \Gamma, (C \wedge D) \Rightarrow B \vdash D
         ~~~~~
         \Gamma, B \vdash G
        }
        {\Gamma, (C \wedge D) \Rightarrow B \vdash G}
        {}}$$
$${\small \irule{\Gamma, (C \vee D) \Rightarrow B \vdash C~~~~~\Gamma, B \vdash G}
        {\Gamma, (C \vee D) \Rightarrow B \vdash G}
        {}}$$
$${\small \irule{\Gamma, (C \vee D) \Rightarrow B \vdash D~~~~~\Gamma, B \vdash G}
        {\Gamma, (C \vee D) \Rightarrow B \vdash G}
        {}}$$
$${\small \irule{\Gamma, (C \Rightarrow D) \Rightarrow B, C \vdash D~~~~~\Gamma, B \vdash G}
        {\Gamma, (C \Rightarrow D) \Rightarrow B \vdash G}
        {}}$$
$${\small \irule{\Gamma, (\fa x~ C) \Rightarrow B \vdash C~~~~~\Gamma, B \vdash G}
        {\Gamma, (\fa x~ C) \Rightarrow B \vdash G}
        {\mbox{$\Rightarrow$-left$_\fa$ ~ $x$ not free in $\Gamma, B$}}}$$
$${\small \irule{\Gamma, (\ex x~C) \Rightarrow B \vdash (t/x)C~~~~~\Gamma, B \vdash G}
        {\Gamma, (\ex x~C) \Rightarrow B \vdash G}
        {\mbox{$\Rightarrow$-left$_\ex$}}}$$
}
\begin{figure}[t]
{\scriptsize
\noindent\framebox{\parbox{\textwidth
}{
{
~~~~~~~~~~~~~~$
\begin{array}{ll}
\irule{}
      {\Gamma, P \vdash P}
      {\mbox{axiom $P$ atomic}}
~~~~~~~~~~~~~~~~~~~~~~~~~~~~~~~~~~~~~~~~~~\\
\irule{}
      {\Gamma \vdash \top}
      {\mbox{$\top$-right}}\\
\irule{}
      {\Gamma, \bot \vdash G}
      {\mbox{$\bot$-left}}\\
\irule{\Gamma, A, B \vdash G}
      {\Gamma, A \wedge B \vdash  G}
      {\mbox{$\wedge$-left}}\\
\irule{\Gamma \vdash A~~~\Gamma \vdash B}
      {\Gamma \vdash A \wedge B}
      {\mbox{$\wedge$-right}}\\
\irule{\Gamma, A \vdash G ~~~ \Gamma, B \vdash G}
      {\Gamma, A \vee B \vdash G}
      {\mbox{$\vee$-left}}\\
\irule{\Gamma \vdash A}
      {\Gamma \vdash A \vee B}
      {\mbox{$\vee$-right}}\\
\irule{\Gamma \vdash B}
      {\Gamma \vdash A \vee B}
      {\mbox{$\vee$-right}}\\
\irule{\Gamma, P, B \vdash G}
      {\Gamma, P, P \Rightarrow B \vdash G}
      {\mbox{$\Rightarrow$-left$_{\mbox{\scriptsize \em axiom}}$}}
\\
\irule{\Gamma, B \vdash G}
      {\Gamma, \top \Rightarrow B \vdash G}
      {\mbox{$\Rightarrow$-left$_\top$}}
\\
\irule{\Gamma, C \Rightarrow B \vdash C
       ~~~
       \Gamma, D \Rightarrow B \vdash D
       ~~~
       \Gamma, B \vdash G}
      {\Gamma, (C \wedge D) \Rightarrow B \vdash G}
      {\mbox{$\Rightarrow$-left$_\wedge$}}
\\
\irule{\Gamma, C \Rightarrow B, D \Rightarrow B \vdash C~~~\Gamma, B \vdash G}
      {\Gamma, (C \vee D) \Rightarrow B \vdash G}
      {\mbox{$\Rightarrow$-left$_\vee$}}
\\
\irule{\Gamma, C \Rightarrow B, D \Rightarrow B \vdash D~~~\Gamma, B \vdash G}
      {\Gamma, (C \vee D) \Rightarrow B \vdash G}
      {\mbox{$\Rightarrow$-left$_\vee$}}
\\
\irule{\Gamma, D \Rightarrow B, C \vdash D~~~\Gamma, B \vdash G}
      {\Gamma, (C \Rightarrow D) \Rightarrow B \vdash G}
      {\mbox{$\Rightarrow$-left$_\Rightarrow$}}
\\
&\irule{\Gamma, (\fa x~ C) \Rightarrow B \vdash C~~~\Gamma, B \vdash G}
      {\Gamma, (\fa x C) \Rightarrow B \vdash G}
      {\mbox{\parbox{3cm}{$\Rightarrow$-left$_\fa$\\$x$ not free in $\Gamma, B$}}}
\\
&\irule{\Gamma, (\ex x~C) \Rightarrow B \vdash (t/x)C~~~\Gamma, B \vdash G}
      {\Gamma, (\ex x~C) \Rightarrow B \vdash G}
      {\mbox{$\Rightarrow$-left$_\ex$}}
\\
\irule{\Gamma, A \vdash B}
      {\Gamma \vdash  A \Rightarrow B}
      {\mbox{$\Rightarrow$-right}}\\
\irule{\Gamma \vdash A}
      {\Gamma \vdash \fa x A}
      {\mbox{\parbox{3cm}{$\fa$-right\\$x$ not free in $\Gamma$}}}
&
\irule{\Gamma, \fa x A, (t/x)A \vdash G}
      {\Gamma, \fa x A \vdash G}
      {\mbox{contr-$\fa$-left}}\\
\irule{\Gamma, A \vdash G}
      {\Gamma, \ex x~A \vdash G}
      {\mbox{\parbox{3cm}{$\ex$-left\\$x$ not free in $\Gamma, G$}}}\\
\irule{\Gamma \vdash (t/x)A}
      {\Gamma \vdash \ex x~A}
      {\mbox{$\ex$-right}}
\end{array}$
\begin{center}
\caption{The system ${\cal D}$\label{d}}
\end{center}}}}
}
\end{figure}
\VL{Unfortunately, only the rules
$\Rightarrow$-left$_{\mbox{\scriptsize \em axiom}}$ and
$\Rightarrow$-left$_\top$ are introduction rules.  But the rules
corresponding to conjunction, disjunction, and implication can be
simplified in the following way
$${\small \irule{\Gamma, C \Rightarrow B \vdash C
         ~~~~~
         \Gamma, D \Rightarrow B \vdash D
         ~~~~~
         \Gamma, B \vdash G
        }
        {\Gamma, (C \wedge D) \Rightarrow B \vdash G}
        {\mbox{$\Rightarrow$-left$_\wedge$}}}$$
$${\small \irule{\Gamma, C \Rightarrow B, D \Rightarrow B \vdash C~~~~~\Gamma, B \vdash G}
        {\Gamma, (C \vee D) \Rightarrow B \vdash G}
        {\mbox{$\Rightarrow$-left$_\vee$}}}$$
$${\small \irule{\Gamma, C \Rightarrow B, D \Rightarrow B \vdash D~~~~~\Gamma, B \vdash G}
        {\Gamma, (C \vee D) \Rightarrow B \vdash G}
        {\mbox{$\Rightarrow$-left$_\vee$}}}$$
$${\small \irule{\Gamma, D \Rightarrow B, C \vdash D~~~~~\Gamma, B \vdash G}
        {\Gamma, (C \Rightarrow D) \Rightarrow B \vdash G}
        {\mbox{$\Rightarrow$-left$_\Rightarrow$}}}$$
}
\VC{These rules can be simplified leading to the system 
${\cal D}$ (Figure \ref{d}).}
The system ${\cal D}$ plus the contr-$\Rightarrow$-left rule is
obviously sound and complete with respect to the system $\cal K$.  
To prove that the
contr-$\Rightarrow$-left rule can be eliminated, and hence the system
${\cal D}$ also is sound and complete with respect to the system $\cal
K$, we use a method similar to that of \cite{DyckhoffNegri}, and prove
the admissibility of the contr-$\Rightarrow$-left rule\VC{---see the long version of the paper for the full proof}.

This system $\cal D$ gives the decidability of a larger fragment of
Constructive
Predicate Logic containing all connectives, shallow universal and
existential quantifiers---that is quantifiers that occur under no
implication at all---and negative existential quantifiers.  This
fragment contains the prenex fragment of Constructive Predicate Logic,
that itself contains Constructive Propositional Logic.

\section*{Acknowledgements}

This work is supported by the ANR-NSFC project LOCALI (NSFC
61161130530 and ANR 11 IS02 002 01) and the Chinese National Basic
Research Program (973) Grant No. 2014CB340302.

\end{document}